\documentclass[a4paper,10pt,USenglish]{article}
\usepackage{fullpage}
\usepackage{amsmath, amssymb}
\usepackage{wrapfig}

\usepackage[shortlabels]{enumitem}

\usepackage[nocompress]{cite}
\newcommand{\ceil}[1]{\left\lceil{#1}\right\rceil}

\usepackage[dvipsnames,usenames]{xcolor}
\usepackage{graphicx}
\graphicspath{{figures/}} %Setting the graphicspath
\usepackage{algorithmic}

\RequirePackage{hyperref}
\hypersetup{colorlinks=true, urlcolor=Blue, citecolor=Green, linkcolor=BrickRed, breaklinks, unicode}

\newcommand{\polylog}{\ensuremath\operatorname{polylog}}

\newcommand{\Start}{\ensuremath\mathrm{start}}

\newcommand{\Char}{\ensuremath\mathrm{char}}
\newcommand{\Pos}{\ensuremath\mathrm{pos}}
\newcommand{\First}{\ensuremath\mathrm{first}}
\newcommand{\Last}{\ensuremath\mathrm{last}}

\newcommand{\parent}{\ensuremath\mathsf{parent}}

\newcommand{\slab}{\ensuremath\mathrm{slab}}

\newcommand{\Rank}{\ensuremath\operatorname{\mathsf{rank}}}
\newcommand{\Select}{\ensuremath\operatorname{\mathsf{select}}}
\newcommand{\PrefixSelect}{\ensuremath\operatorname{\mathsf{prefix-select}}}

\newcommand{\SegmentSelect}{\ensuremath\operatorname{\mathsf{segment-select}}}
\newcommand{\SlabSum}{\ensuremath\operatorname{\mathsf{slab-sum}}}
\newcommand{\SlabSelect}{\ensuremath\operatorname{\mathsf{slab-select}}}
\newcommand{\Access}{\ensuremath\operatorname{\mathsf{access}}}
\newcommand{\Replace}{\ensuremath\operatorname{\mathsf{replace}}}
\newcommand{\Insert}{\ensuremath\operatorname{\mathsf{insert}}}
\newcommand{\Delete}{\ensuremath\operatorname{\mathsf{delete}}}
\newcommand{\Predecessor}{\ensuremath\operatorname{\mathsf{predecessor}}}

\newtheorem{theorem}{Theorem}

\newtheorem{lemma}[theorem]{Lemma}

\newenvironment{proof}{

\noindent{\bf Proof:}} {\hfill$\blacksquare$\medskip}

\title{Random Access in Persistent Strings and Segment Selection\footnote{An extended abstract appeared at the 31st International Symposium on Algorithms and Computation~\cite{BG2020}}}

\author{Philip Bille \\\texttt{phbi@dtu.dk} \and Inge Li G{\o}rtz \\\texttt{inge@dtu.dk}}

\begin{document}

\maketitle
\begin{abstract}
	We consider compact representations of collections of similar strings that support random access queries. The collection of strings is given by a rooted tree where edges are labeled by an edit operation (inserting, deleting, or replacing a character) and a node represents the string obtained by applying the sequence of edit operations on the path from the root to the node. The goal is to compactly represent the entire collection while supporting fast random access to any part of a string in the collection. This problem captures natural scenarios such as representing the past history of an edited document or representing highly-repetitive collections. Given a tree with $n$ nodes, we show how to represent the corresponding collection in $O(n)$ space and $O(\log n/ \log \log n)$ query time. This improves the previous time-space trade-offs for the problem. Additionally, we show a lower bound proving that the query time is optimal for any solution using near-linear space. 
	
To achieve our bounds for random access in persistent strings we show how to reduce the problem to the following natural geometric selection problem on line segments. Consider a set of horizontal line segments in the plane. Given parameters $i$ and $j$, a segment selection query returns the $j$th smallest segment (the segment with the $j$th smallest $y$-coordinate) among the segments crossing the vertical line through $x$-coordinate $i$. The segment selection problem is to preprocess a set of horizontal line segments into a compact data structure that supports fast segment selection queries. We present a solution that uses $O(n)$ space and support segment selection queries in $O(\log n/ \log \log n)$ time, where $n$ is the number of segments. Furthermore, we prove that that this query time is also optimal for any solution using near-linear space.
\end{abstract}

\section{Introduction}
The \emph{random access problem} is to preprocess a data set into a compressed representation that supports fast retrieval of any part of the data without decompressing the entire data set. The random access problem is a well-studied problem for many types of data and compression schemes~\cite{BLRSSW2015, VY2013, BCPT2014, FV2007, SG2006, GRRV2013, MN2015, BCGNN2014, BGLW2015, BCCGSVV2018, KP2018} and random access queries is a basic primitive in several algorithms and data structures on compressed data, see e.g.,~\cite{BLRSSW2015,GGKNP2012,GGKNP2014,BEGV2018,GGP2015}

In this paper, we initiate the study of the random access problem on collections of strings where each string is the result of an \emph{edit operation}, i.e., insert, delete, or replace a single character, from another string in the collection. Specifically, our collection is given by a rooted tree, called a \emph{version tree}, where edges are labeled by an edit operation, the root represents the empty string, and a node represents the string obtained by applying the sequence of edit operation on the path from the root to the node (see Figure~\ref{fig:reduction}(a)). We call such a collection a \emph{persistent string} since we can naturally view it as persistent versions of a single string. Given a node $v$ and an index $j$, a random access query returns the character at position $j$ in the string represented by $v$.

Random access in persistent strings captures natural scenarios for collections of similar strings. For instance, consider the problem storing and accessing the past history of edits in a document. Instead of explicitly storing all versions of the document, we can represent the entire history compactly as a path of updates. Random access in a past version of the document then corresponds to a random access query on the corresponding node on the path. In our setup we can even support branching in the history of the document, as in version control systems, to form a tree of document histories. As another example, consider storing and accessing a collection of related genome sequences. If we know (a good approximation of) the edit distance between the pairs of genome sequences, we can construct a small version tree representing the collection from the minimum spanning tree of the pairs of distance. Again, random access in a sequence in the collection corresponds to a random access query on the corresponding node.

To the best our knowledge, no previous work has explicitly considered random access on persistent strings, but several well-known techniques and results can be combined to provide non-trivial bounds on the problem (we review these solutions in Section~\ref{sec:previouswork}). However, all of these solutions lead to suboptimal bounds. In this paper, we introduce a new representation of persistent strings that supports random access. Our representation uses $O(n)$ space and supports random access queries in $O(\log n/ \log \log n)$ time, where $n$ is the number of nodes in the version tree (or equivalently the number of strings in the collection). This improves the best known combinations of time and space among all previous solutions. Furthermore, we prove that any solution using near linear space needs $\Omega(\log n/\log \log n)$ query time, thus showing that our query time is optimal.

To achieve our bounds for random access in persistent strings we show how to reduce the problem to the following natural geometric selection problem on line segments. %of independent interest. 
Consider a set of horizontal line segments in the plane. Given parameters $i$ and $j$, a \emph{segment selection query} returns the $j$th smallest segment (the segment with the $j$th smallest $y$-coordinate) among the segments crossing the vertical line through $x$-coordinate $i$. The \emph{segment selection problem} is to preprocess a set of horizontal line segments into a compact data structure that supports fast segment selection queries. To the best of our knowledge no previous results are known for segment selection. In this paper, we present a solution that uses $O(n)$ space and support segment selection queries in $O(\log n/ \log \log n)$ time, where $n$ is the number of segments. Furthermore, by combining the lower for random access in persistent strings and our reduction we show that our query time for segment selection is also optimal for any solution using near linear space.

\subsection{Random Access in Persistent String}
Let $T$ be a version tree with $n$ nodes. Each node $v$ of $T$ represents a string $S(v)$ and each edge is labeled by one of the following edit operations: 
\begin{itemize}
	\item $\Replace(k, \alpha)$: change the $k$th character to $\alpha$. 
	\item $\Insert(k,\alpha)$: insert character $\alpha$ immediately after position $k$.
	\item $\Delete(k)$: delete the character at position $k$.
\end{itemize}
The string represented by the root is the empty string $\varepsilon$, and the string represented by a non-root node $v$ is the result of applying all edit operations on the path from the root to $v$ on the empty string. Our goal is to preprocess $T$ into a compact data structure that supports the query $\Access(v,j)$, that returns $S(v)[j]$. While no previous work has explicitly considered random access in persistent strings, standard techniques can be adapted to achieve non-trivial time-space trade-offs. In particular, using persistent binary search trees leads to a solution with $O(n)$ space and $O(\log n)$ query time. Alternatively, using recent grammar compression techniques to represent the collection leads to a solution with $O(n\log^{1+\varepsilon} n)$ space and $O(\log n/ \log \log n)$ time. We review these solutions in Section~\ref{sec:previouswork}. We present a new representation of persistent strings that achieves the following bound: 	
\begin{theorem}\label{thm:mainupperbound} 
Given a version tree $T$ with $n$ nodes we can solve the random access problem in $O(n)$ space and $O(\log n/ \log \log n)$ query time. Furthermore, we can report a substring of length $\ell$ using $O(\ell)$ additional time. 
\end{theorem}
Theorem~\ref{thm:mainupperbound} simultaneously matches the best known space and time bounds of the previous approaches. In particular, compared to the solution using binary search trees we match the space while improving the $O(\log n)$ query time to $O(\log n/\log \log n)$. On the other hand, compared to the solution using grammar compressed techniques, we match the query time while improving the space from $O(n\log^{1+\varepsilon} n)$ to linear. Furthermore, we show the following matching lower bound. 
\begin{theorem}\label{thm:mainlowerbound} 
Any data structure that solves the random access problem on a version tree $T$ with $n$ nodes using $n \log^{O(1)} n$ space needs $\Omega(\log n/ \log \log n)$ query time. This holds even in the special case when $T$ is a path.  
\end{theorem}
Hence, the query time in Theorem~\ref{thm:mainupperbound} is optimal for any near-linear space solution. Note that Theorem~\ref{thm:mainlowerbound} even holds for version trees that are simple paths, such as in the example with storing and accessing the past history of edits in a document.

\subsection{Segment Selection}
Let $L$ be a set of $n$ horizontal line segments in the plane. The \emph{segment selection problem} is to preprocess $L$ to support the operation: 
\begin{itemize}
	\item $\SegmentSelect(i,j)$: return the $j$th smallest segment (the segment with the $j$th smallest $y$-coordinate) among the segments crossing the vertical line through $x$-coordinate $i$.  
\end{itemize}
To the best our knowledge no previous work has considered the segment selection problem. A number of related problems on orthogonal line segments are well-studied. For instance, in the 1-D stabbing max problem, the goal is to store a set of horizontal line segments, each with a given priority such that we can quickly return the segment of highest priority crossing the vertical line through $x$-coordinate $i$, see e.g.,~\cite{Nekrich2011,AAKMTY2012,CT2018}. Another related problem is the 1-D vertical ray shooting problem. Here, the goal is to store a set of horizontal line segments such that given a query point $q$ we can quickly return the lowest segment above $q$, see e.g.,\cite{BVS1995,CP2009,Chan2013,CT2018}. We view segment selection as a natural variant and believe that it will likely be of independent interest. We present a new representation that achieves the following bounds.
\begin{theorem}\label{thm:segmentselection} 
Given a set of $n$ horizontal segments in the plane, we can solve the segment selection problem in $O(n)$ space and $O(\log n/ \log \log n)$ query time.
\end{theorem}
As an direct consequence of Theorem~\ref{thm:mainlowerbound} and our reduction from random access on a persistent string, we obtain the following lower bound for segment selection.   
\begin{theorem}\label{thm:segmentselectionlowerbound}
	Any data structure that solves the segment selection problem on $n$ segments nodes using $n \log^{O(1)} n$ space needs $\Omega(\log n/ \log \log n)$ query time. 
\end{theorem}
Hence, the query time in Theorem~\ref{thm:segmentselection} is optimal for any near-linear space solution.

\subsection{Techniques}
To achieve our result we introduce several techniques of independent interest. %
First, we show how to reduce random access queries on a persistent string to segment selection queries. The main idea is to traverse the version tree in a depth-first traversal and produce segments representing characters appearing in the versions of the persistent strings. The $x$-coordinates of the segments correspond to the traversal time interval and the $y$-coordinates correspond to the left-to-right ordering of the characters in the strings. We show how to construct segments such that at any point in time $i$, the segments crossing the vertical line through $x$-coordinate $2i$ corresponds to the string represented at the node in $T$ first visited at time $i$. It follows that any random access query can be answered by a corresponding segment selection query.

Next, we show how to efficiently solve the segment selection problem in linear space and $O(\log n/ \log \log n)$ query time. To do so, the main idea is to build a balanced tree of degree $\Delta = O(\log^\varepsilon n)$ and of height $O(\log_\Delta n) = O(\log n/\log \log n)$ that stores the segments ordered by $y$-coordinate. Each internal node thus partitions the segments below it into $\Delta$ horizontal bands called \emph{slabs}. 

To answer a segment selection query $(i,j)$ we traverse the tree to find the leaf containing the $j$th segment that crosses the vertical line at time $i$. To implement the traversal we need to determine at each node $v$ the slab containing the desired segment among the segments below $v$ at the specified time $i$. The key challenge is to compactly represent the segments while achieving constant query time to find the correct slab at each node. Using well-known  techniques we can solve this \emph{slab selection problem} with an explicit representation of segments below $v$ in constant time and $O(n_v)$ of space, where $n_v$ is the number of segments below $v$. Unfortunately, this leads to a solution to segment selection that uses $O(n \log_\Delta n) = O(n \log n/ \log \log n)$ space. We show how to compactly represent the segments to significantly improve the space to $O(n_v \log \log n)$ \emph{bits} while simultaneously achieving constant time queries. In turn, this implies a solution to segment selection using $O(n)$ space and $O(\log_\Delta n) = O(\log n/\log \log n)$ query time.

Finally, we prove a matching lower bound for the random access in persistent strings problem by showing that any solution using $n \log^{O(1)} n$ space needs $\Omega(\log n/ \log \log n)$ query time. To do so we show a simple reduction from the \emph{range selection problem}~\cite{JL2011}. By the reduction from random access queries in persistent strings to segment selection queries it directly follows that the same lower bound applies to segment selection.

\subsection{Previous Work}\label{sec:previouswork} 
To the best of our knowledge no previous work has explicitly considered supporting random access in persistent strings. However, several existing approaches can be applied or extended to obtain non-trivial solutions to the problem and several related models of repetitiveness have been proposed. We discuss these in the following. To state the bounds, let $T$ be a version tree with $n$ nodes representing a collection of $n$ strings of total size $N$. Since any string represented by a node in $T$ can be the result of at most $n$ insertions we have that $N = O(n^2)$. Hence, naively we can solve the random access problem by explicitly storing all strings using $O(N) = O(n^2)$ space and $O(1)$ query time. With techniques from either persistent or compressed data structures we can significantly improve this as discussed below.  

\paragraph{Persistent Data Structures and Dynamic Strings} 
Ordinary data structures are \emph{ephemeral} in the sense that updating the data structure destroys the old version and only leaves the new version. A data structure is \emph{persistent} if it preserves old versions of itself and allows queries and/or updates to them. In \emph{partial persistence} we allow queries on all versions but only updates on the newest version, and in \emph{full persistence} we allow queries and updates on all versions. Thus, in partial persistence the versions form a path whereas in full persistence the versions form a tree called the \emph{version tree}. Persistent data structures is a classic data structural concept and were first formally studied by~Driscoll et. al.~\cite{JSST1989}.

A \emph{dynamic string} data structure supports the edit operations (insert, delete, and replace) and access to any character in the string. An immediate approach to solve the random access problem in persistent strings is to make a dynamic string data structure fully persistent. To do so, we simply traverse the version tree and perform the edit operations on the edges. To answer a random access query on a string represented by a node $v$ we simply perform a persistent access operation on the version of the data structure corresponding to version $v$. Depending on the dynamic string data structure we obtain different time-space trade-offs for the random access problem. A balanced binary search tree implements a dynamic string data structure using $O(\log n)$ time for all operations. Since binary search trees are constant degree pointer data structures a classic transformation by Driscoll et al.~\cite{JSST1989} immediately implies an $O(\log n)$ time solution for access. Since each persistent update to the binary search tree incurs  $O(\log n)$ space overhead this leads to a total space of $O(n \log n)$. With a more careful implementation of binary search trees the space can be improved to $O(n)$~\cite{JSST1989, ST86}.

Maintaining a dynamic string (often called the list representation or list indexing problem~\cite{FS1989, Dietz1989}) is well-studied and closely connected to the partial sums problem. Dietz~\cite{Dietz1989} presented the first solution achieving $O(\log n/ \log \log n)$ time for access and updates  and Fredman and Saks~\cite{FS1989} showed in their seminal paper on cell probe complexity that this bound is optimal. Several variations and extension have been proposed~\cite{BCPS2017,BCCGSVV2018,HSS2011,PT2014,RRR2001,PD2006,Fenwick1994}. However, all of these solutions rely on word RAM techniques and therefore incur an overhead of $\Omega(\log \log n)$ time to make them persistent~\cite{Die89} thus leading to a solution to the random access problem with query time $\Theta(\log n)$.

\paragraph{Compressed Representations}
The classic Lempel-Ziv compression scheme (LZ77)~\cite{ZL1977} compresses an input string $S$ by parsing $S$ into $z$ substrings  $f_1f_2\ldots f_z$, called \emph{phrases}, in a greedy left-to-right order. Each phrase is either the first occurrence of a character or the longest substring that has at least one occurrence starting to the left of the phrase. By replacing each phrase by a reference to the previous occurrences we obtain a compressed representation of the string of length $O(z)$. 

We can use LZ77 compression to efficiently store all versions of the persistent string in the random access problem. To do so, we write all the strings represented in the version tree $T$ and concatenate them in order of increasing depth in $T$. The string represented by a node $v$ can be formed from the string of the parent of $v$ by at most $3$ substrings, namely, the substrings before and after the edit operation and a new character in case of a replace or insert operation. Since we concatenate the strings in increasing depth it follows that the greedy LZ77 parsing uses at most $z = O(n)$ phrases.

To solve the random access problem on the persistent string we can convert the LZ77 compressed representation into a small grammar representation and then apply efficient random access results for grammars. Converting the LZ77 compressed string leads to a grammar of size $O(z \log (N/z)) = O(n \log n)$~\cite{CLLPPSS2005, Rytter2003}. Using the best known trade-offs for random access in grammars, this  leads to solutions using either $O(n\log n)$ space and $O(\log N) = O(\log n)$ query time~\cite{BLRSSW2015} or $O(n\log^{1+\varepsilon} n)$ space and $O(\log N/\log \log N) = O(\log n/ \log \log n)$ query  time~\cite{BCPT2014,GZL2018}. We note that both of these results inherently need superlinear space for the conversion from LZ77 to grammars~\cite{CLLPPSS2005}. Furthermore, Verbin and Yu~\cite{VY2013} showed that the latter query time is optimal. More precisely, they proved that any representation of an LZ77 compressed string using $z \log^{O(1)} N = n \log^{O(1)} n$ space must use $\Omega(\log N/\log \log N) = \Omega(\log n/ \log \log n)$ time.

A related simpler model of compression is \emph{relative compression}~\cite{SS1978,Storer1982} (see also~\cite{COMNVW2012, DJSS2014, HPZ2011, KPZ2010, KPZ2011, LDKTW1998,LDK1999,BCCGSVV2018}), where we explicitly store a single reference string and compress a collection of strings as substrings of the reference string. A similar compression model is also proposed in~\cite{GKNPS2013,MNSV2010,Navarro2012,Navarro2019}. The relative compression model compresses efficiently if each string is the result of applying a small number of edits to the base string. In contrast, using persistent strings we can compress efficiently if each string is the result of editing \emph{any} other string in the collection. 

\subsection{Outline}
We present the reduction from random access to segment selection in Section~\ref{sec:reduction} and our solution to the slab selection problem in Section~\ref{sec:slabselection}.
%In Section~\ref{sec:reduction} we %first 
%present the reduction from random access to segment selection. 
%We present our solution to the slab selection problem in Section~\ref{sec:slabselection}. 
We then use our slab selection data structure in our full data structure for the segment selection problem in Section~\ref{sec:segmentselection}. Plugging this into our reduction leads to Theorem~\ref{thm:mainupperbound}. We show the lower bounds in Section~\ref{sec:lowerbound} and conclude with some open problems in Section~\ref{sec:conclusion}.

\section{Reducing Random Access to Segment Selection}\label{sec:reduction} 
In this section we show how to reduce the random access problem to the following natural geometric selection problem on line segments. Let $L$ be a set of $n$ horizontal line segments in the plane. The \emph{segment selection problem} is to preprocess $L$ to support the operation: 
\begin{itemize}
	\item $\SegmentSelect(i,j)$: return the $j$th smallest segment (the segment with the $j$th smallest $y$-coordinate) among the segments crossing the vertical line through $x$-coordinate $i$.  
\end{itemize}
We will view  the $x$-axis as a timeline and  often refer to an $x$-coordinate $i$ as time $i$. We will show how to efficiently solve the segment selection problem in the following sections. Our reduction from the random access problem works as follows. Let $T$ be an instance of the random access problem with $n$ nodes and assume wlog. that $T$ contains no edges labeled by $\Replace$. We can do so since we can always convert edges labeled by $\Replace$ into two edges labeled by a $\Delete$ and $\Insert$, thus at most doubling the size of the instance. We construct an instance $L$ of segment selection as follows.

\begin{figure}[t]
\centering
  \includegraphics[scale=0.43]{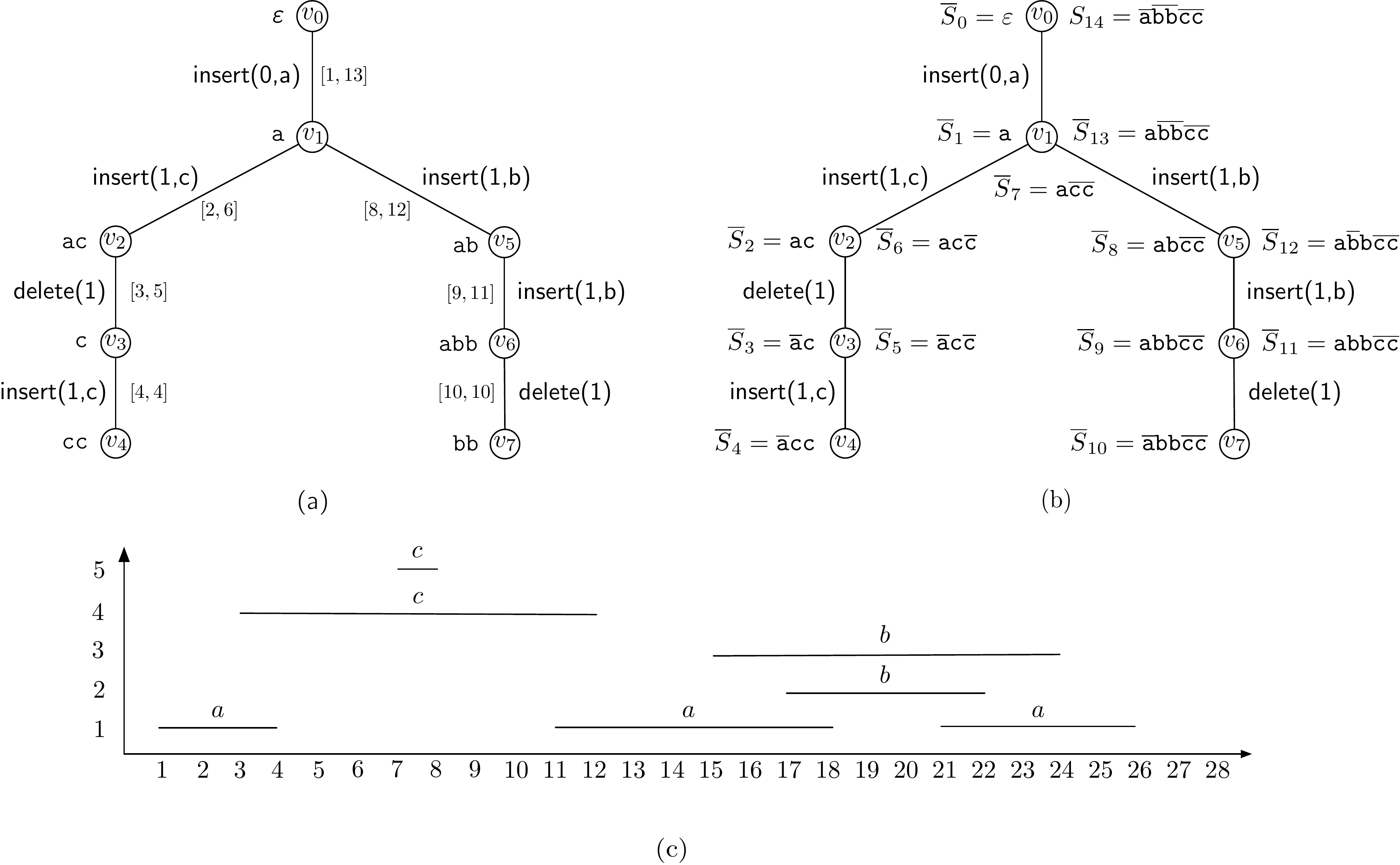}
  \caption{\label{fig:reduction} (a) A persistent string representing the collection $\{\texttt{$\epsilon$}, \texttt{a}, \texttt{ac}, \texttt{c}, \texttt{cc}, \texttt{ab}, \texttt{abb}, \texttt{bb}\}$. The interval $I(e)$ is shown for each edge. (b) The marked strings of (a). The insertion edges are unmarked in the following intervals: $(v_0,v_1)$ in  $[1,2] \cup [6,9]\cup [11,13]$, $(v_1, v_2)$ in $[2,6]$, $(v_3, v_4)$ in $[4,4]$, $(v_1, v_5)$ in $[8,12]$, and $(v_5,v_6)$ in $[9,11]$.  
% The corresponding intervals are $a: [1,2] \cup [7,9]\cup [12,14]$, $b: [9,12]$, $b: [8,13]$, $c: 2,7]$, and $c: [4,5]$.
  (c) The segment selection instance corresponding to (a). The range of  $x$-coordinates of segments are obtained by converting each interval $[i,j]$ above to $[2i-1, 2j]$. 
%  {\bf Inge: Change segments to 2i-1 and 2j: a: [1,4],[13,19],[23,24]. b: [17,24]. b: [15, 26]. c: [3, 14]. c[7, 10]. Should we also add intervals??? }
}   
\end{figure}

We first perform an Euler tour~\cite{TV1984} of $T$ to construct a sequence $\overline{S}_0, \ldots, \overline{S}_{2n-2}$ of strings corresponding to each time we meet a node in the Euler tour.  We call these strings \emph{marked strings} since each character in them will be either \emph{marked} or \emph{unmarked}. The marked strings are defined as follows.
String $\overline{S}_0$ is the empty string. Suppose we have constructed  $\overline{S}_0, \ldots, \overline{S}_{\ell-1}$ and let $e$ be the edge visited at time $\ell$ in the Euler tour. We construct $\overline{S}_\ell$ from $\overline{S}_{\ell-1}$ according to the following cases (see Figure~\ref{fig:reduction}(b) for an example).

%We first perform a depth-first traversal to construct a sequence  $\overline{S}_1, \ldots, \overline{S}_{2n}$ corresponding to the $2n$ start and finish times at each node in the traversal. We call these strings \emph{marked strings} since each character will be either \emph{marked} or \emph{unmarked}. We construct these as follows. String $\overline{S}_1$ is the empty string. Suppose we have constructed  $\overline{S}_1, \ldots, \overline{S}_{\ell-1}$ and consider the edges in the traversal between the nodes visited at time $\ell-1$ and $\ell$. There can be either one or two edges on this path (one if one of the nodes is a  parent of the other or two if the nodes are siblings). For each of these edges $e$ in order of the traversal we do the following. 

%and let $e$ be an edge on the path connecting the nodes in $T$ visited at time $\ell-1$ and $\ell$. We construct $\overline{S}_\ell$ from $\overline{S}_{\ell-1}$ according to the following cases (see Figure~\ref{fig:reduction}(a)).
\begin{description}
	\item[Case 1: Insertions] Suppose that $e$ is labeled $\Insert(i,\alpha)$. If we traverse $e$ in the downward direction, we insert character $\alpha$ as an unmarked character in $\overline{S}_{\ell-1}$ immediately to the right of the $i$th unmarked character to get $\overline{S}_\ell$. If we traverse $e$ in the upwards direction we mark the same character that was inserted as an unmarked character in the earlier downwards traversal of $e$.
	\item[Case 2: Deletions] Suppose that $e$ is labeled $\Delete(i)$. If we traverse $e$ in the downward direction, we mark the $i$th unmarked character in $\overline{S}_{\ell-1}$ to get $\overline{S}_{\ell}$. If we traverse $e$ in the upward direction, we unmark the same character that was marked in the downward traversal of $e$. 
\end{description}
Note that an insertion edge $e$ traversed in the downward direction at time $\ell$ results in an insertion of a character, denoted $\Char(e)$, in $\overline{S}_\ell$. Since $\Char(e)$ is never removed from subsequent marked strings it appears in all subsequent strings $\overline{S}_{\ell}, \ldots, \overline{S}_{2n-2}$, but changes between being marked and unmarked. If a deletion edge $e'$ changes $\Char(e)$ from unmarked to marked we say that $e'$ \emph{deletes} $\Char(e)$. 

%For a node $v$ in $T$, let $\Start(v)$ and $\Finish(v)$ denote the start and finish times, respectively, in the traversal of $T$, and for an edge $e = (v, \parent(v))$ let $I(e) = [\Start(v), \Finish(v)]$ denote the time interval of the nodes in the subtree of the bottom endpoint $v$ of $e$. 

%For a node $v$ in $T$, let $\First(v)$ and $\Last(v)$ denote the first and last time, respectively, we visit $v$ in the traversal of $T$, and for an edge $e = (v, \parent(v))$ let $I(e) = [\First(e), \Last(e)]$ denote the time interval of the nodes in the subtree of the bottom endpoint $v$ of $e$.

For an edge $e$ in $T$, let $\First(e)$ and $\Last(e)$ denote the first and last time, respectively, we visit $v$ in the  Euler tour of $T$, and let $I(e) = [\First(e), \Last(e)-1]$ denote the interval of $e$. %if $e$ is %an insertion edge and $I(e) = [\First(e), \Last(e)-0.5]$ if $e$ is a deletion edge. 

\begin{lemma}\label{lem:markedstring} 
Let $e$ be an insertion edge in $T$ that is traversed in the downward direction at time $\ell$ and let $e_1, \ldots, e_m$ be the edges in $T(v)$ that delete $\Char(e)$. Then, $\Char(e)$ is unmarked in all strings $\overline{S}_i$ where $i$ is an integer in the interval $I(e) \setminus \left(I(e_1) \cup \cdots I(e_m) \right)$ and  marked in $\overline{S}_i$ for all other integers $i$ in $[\ell, 2n-2]$.
\end{lemma}
\begin{proof}
    We have that $\Char(e)$ appears in $S_\ell, \ldots, S_{2n-2}$. The edge $e$ inserts $\Char(e)$ as unmarked in the interval $I(e)$ and each edge $e'$ that deletes $\Char(e)$, marks it in the interval~$I(e')$. 
\end{proof}
For instance, consider $e = (v_0, v_1)$ in Figure~\ref{fig:reduction}(a) that inserts an \texttt{a} which is then deleted by $e_1 = (v_3, v_2)$ and $e_2 = (v_7, v_6)$. Thus, \texttt{a} appears in the interval $[1, 13]$ and is unmarked in $I(e) \setminus (I(e_1) \cup I(e_2)) = [1,13] \setminus ([3,5] \cup [10,10]) = [1,2] \cup [6,9] \cup [11, 13]$. 

For a node $v$ in $T$, let $\Start(v)=\First((\parent(v),v))$
denote the first time we meet $v$ in the Euler tour of $T$. For the root $r$ we define $\Start(r)=0$.
   
\begin{lemma}\label{lem:markedstringrelation}
	For any $v$, the concatenation of the unmarked characters in $\overline{S}_{\Start(v)}$ is $S(v)$. 
\end{lemma}
\begin{proof} 
From the Lemma~\ref{lem:markedstring}, the unmarked characters in $\overline{S}_{\Start(v)}$ are those which have been inserted at an edge $(w, \parent(w))$ where $w$ is ancestor of $v$ and have not been marked by any deletion edge in between. By definition these are the same characters as $S(v)$. From the insertion ordering of the characters in the marked strings it follows that characters in     $\overline{S}_{\Start(v)}$ and $S(v)$ appear in the same order.  	
\end{proof}
Next, we construct a set of labeled line segments $L$ from $\overline{S}_{2n-2}$ as follows. Note that $\overline{S}_{2n-2}$ consists of all of the (marked) characters appearing at insertion edges in $T$. For each insertion edge $e$, define $\Pos(e)$ to be the position of $\Char(e)$ is $\overline{S}_{2n-2}$. For instance, in Figure~\ref{fig:reduction}(a) $\Pos((v_1, v_0)) = 1$ since \texttt{a} is at position $1$ in $\overline{S}_{14}$. For each insertion edge $e$ in $T$ that is deleted by edges $e_1, \ldots, e_m$, we construct $m+1$ horizontal line segments corresponding to the $m+1$ time intervals where $\Char(e)$ is unmarked. These $m+1$ segments are all labeled by $\Char(e)$ and all have $y$-coordinate $\Pos(e)$.  For an interval $[i,j]$ the corresponding segment has $x$-coordinates $2i-1$ and $2j$. We use $2i-1$ and $2j$ to ensure that all segments have length at least one and that no two segments share an endpoint. See Figure~\ref{fig:reduction}(b). For instance, the insertion edge $e = (v_0, v_1)$ has position $1$ and two deletion edges producing the $3$ segments in Figure~\ref{fig:reduction}(b) labeled \texttt{a}. We have the following correspondence between $T$ and $L$. %For each insertion edge $e$ in $T$ that is deleted by edges $e_1, \ldots, e_m$, we construct $m+1$ horizontal line segments. The segments are all labeled by $\Char(e)$, all have $y$-coordinate $\Pos(e)$, and have $x$-coordinates corresponding to the $m+1$ time intervals where $\Char(e)$ is unmarked. See Figure~\ref{fig:reduction}(b). For instance, the insertion edge $e = (v_0, v_1)$ has position $1$ and two deletion edges producing the $3$ segments in Figure~\ref{fig:reduction}(b) labeled \texttt{a}. We have the following correspondence between $T$ and $L$. 

\begin{lemma}\label{lem:segments} Let $T$ be a version tree and let $L$ be the corresponding instance of the segment selection. Then, $S(v)$ is the concatenation labels of the segments crossing the vertical line at time $2\cdot \Start(v)$ ordered by increasing $y$-coordinate.
\end{lemma}
\begin{proof}
	We first show that the vertical line at $2\cdot \Start(v)$ crosses exactly the segments corresponding to unmarked characters in $\overline{S}_{\Start(v)}$. By the definition of the intervals and the segments it is enough to show that $i\leq  \Start(v) \leq j$ if and only if $2i-1\leq  2\cdot \Start(v) \leq 2j$. This follows immediately from the fact that $i$, $j$, and $\Start(v)$ are integers. By the definition of $\Pos(e)$ the order of the segments is the same as the order of the corresponding unmarked characters in $\overline{S}_{\Start(v)}$. Thus the segments crossing the vertical line at time $2\cdot \Start(v)$ in increasing order is the concatenation of the unmarked characters in $\overline{S}_{\Start(v)}$. By Lemma~\ref{lem:markedstringrelation} this is $S(v)$.
\end{proof}

Each edge in $T$ increases the number of segments in $L$ by at most $1$ and hence $L$ contains at most $n-1$ segments. To answer $\Access(v,j)$ on $T$ we compute $\SegmentSelect(2\cdot  \Start(v), j)$ on $L$ and return the corresponding label. By Lemma~\ref{lem:segments} this correctly returns $S(v)[j]$. In summary, we have the following result. 

\begin{lemma}\label{lem:reduction}
Given a solution to the segment selection problem on $n$ segments that uses $s(n)$ space and answers queries in $t(n)$ time, we can solve the random access problem in $O(s(n))$ space and $O(t(n))$ time. 
\end{lemma}

\section{Selection in Slabs}\label{sec:slabselection}
In this section, we introduce the \emph{slab selection problem} and present an efficient solution. Our data structure will be a key component in our full solution to the segment selection problem that we present in the next section.  As before we will view the $x$-axis as a timeline and often refer to an $x$-coordinate $i$ as time $i$.
%\begin{itemize}
%\item We assume that the $x$-coordinates of the set of the start and end points of the $n$ segments is the set $\{1,\ldots, 4n\}$, and that at every time $i$ at most one segments starts or at most one segments ends. 
%\item
%Note that this is the case in the reduction from Section~\ref{sec:reduction}.
%\item $i$ and $j$ in the queries are integers. 
%\item Make clear slabs are of the same size.
%\end{itemize}

Let $L$ be a set of $n$ segments given in the following "rank reduced" coordinates. The $x$-coordinates of the segment endpoints are unique integers from the set $\{1, \ldots, 4n\}$ and the $y$-coordinates are unique integers in $\{1,\ldots, n\}$. In particular, at every time at most one segment starts or ends. Note that the condition on the $x$-axis is satisfied in the reduction from Section~\ref{sec:reduction}. To satisfy the condition on the $y$-axis, we sort the segments according to their $y$-coordinate breaking ties according to their starting point on the $x$-axis, and use their rank in this ordering as $y$-coordinate. Note that this maintains the ordering among segments crossing the vertical at any time $i$. 

We partition the segments $L$ into $\Delta = O(\log^\varepsilon n)$, where $0 < \varepsilon < 1$, infinite horizontal bands $s_1, \ldots, s_\Delta$, called \emph{slabs}. Each slab consists of $\ceil{n/\Delta}$ segments, except possibly $s_\Delta$ which may be smaller. The \emph{slab selection problem} is to compactly represent $L$ to support the following queries: 
\begin{itemize}
	\item $\SlabSum(i, j)$: return the total number of segments in slabs $s_1, \ldots, s_j$ crossing the vertical line through $x$-coordinate $i$. 
	\item $\SlabSelect(i,j)$: return the smallest $k$ such that $\SlabSum(i,k) \geq j$. 
\end{itemize}
%\pbcom{Add picture of slabs?}
The goal of this section is to construct a data structure for the slab selection problem that uses $O(n \log \log n)$ bits of space and answers $\SlabSum$ and $\SlabSelect$ queries in constant time. Note that if we explicitly represent each of the $n$ segments, e.g., by their two $x$-coordinate endpoints and their $y$-coordinate, we need $\Omega(n \log n)$ bits even if we ignore how to support queries. We present a compact representation of the collection of segments that improves the space to $O(n \log \log n)$ bits and simultaneously achieves constant time queries.

\begin{figure}[t]
\centering
  \includegraphics[scale=0.45]{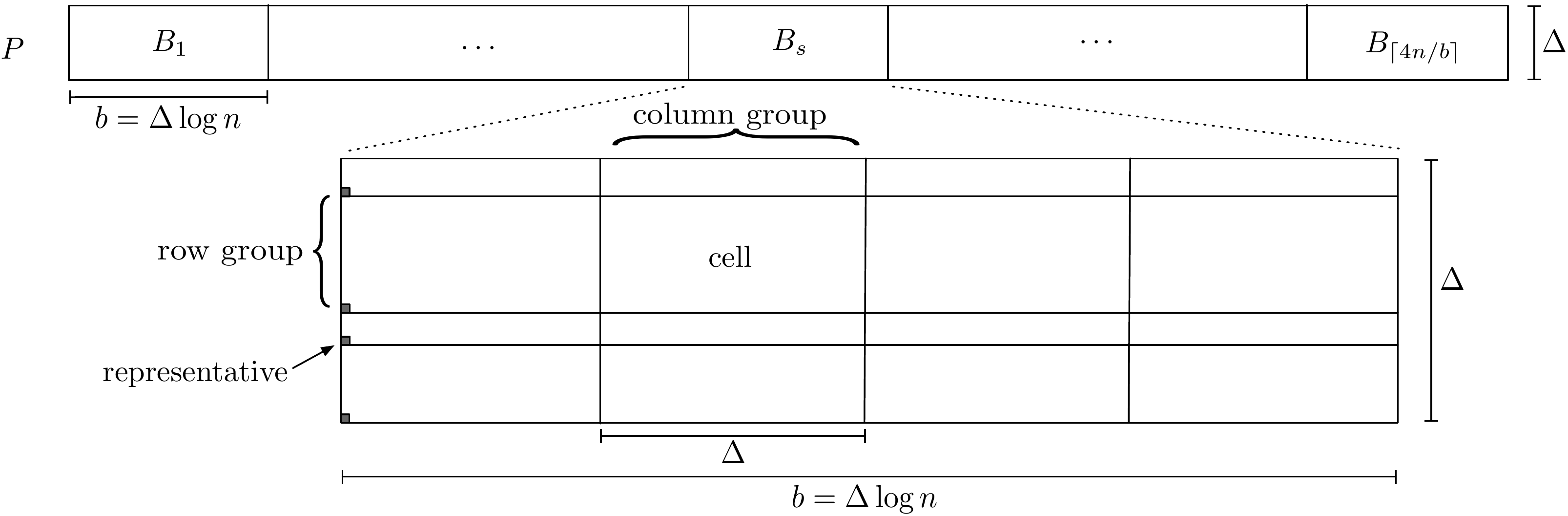}
  \caption{\label{fig:grid} The grid $P$ partitioned into blocks and a block of $P$ partitioned into column groups, row groups, and cells.}   
\end{figure}

Before presenting our data structure, we first convert the problem to a problem on a grid of prefix sums, define a decomposition on the grid, and show some key properties that we will need in our solution. 
 We define a grid $P$ of integers arranged in $4n$ columns and $\Delta$ rows such that the entries in column $i$ represent the prefix sums of the number segments crossing at time $i$. % in each of $\Delta$ slabs.
We use $P(i,j)$ to denote the entry in column $i$ and row $j$ in $P$. More precisely, $P(i,j)$ contains the number of segments crossing $i$ in slab $s_1$ to $s_j$.
We have that $\SlabSum(i,j) = P(i,j)$ and $\SlabSelect(i,j)$ corresponds to a predecessor query on column $i$, that is, computing the smallest $k$ such that $P(i,k) \geq j$.  

We decompose $P$ as follows. Let $b = \Delta \log n$. We partition $P$ into \emph{blocks} $B_1, \ldots, B_{\ceil{4n/b}}$ of $b$ consecutive columns. We further partition each block $B$ into groups of consecutive columns and rows called \emph{column groups} and \emph{row groups}, respectively (see Figure~\ref{fig:grid}). The column groups are groups of $\Delta$ consecutive columns and the row groups are defined such that two adjacent rows are in the same row group if their leftmost entries differ by at most $b$. Each rectangular subgrid in $B$ given by the entries that are in the same column group and row group is called a \emph{cell} of $B$. The \emph{representative} of a row group in $B$ is the bottom and leftmost position in the row group. The \emph{representative} of a cell $C$ in $B$ is the representative of the row group of $C$. For any cell $C$ in $B$, we define the \emph{normalized cell}, denoted $\hat{C}$, to be $C$ where all entries have been subtracted by the representative of $C$. We have the following properties of the construction. 

\begin{lemma}\label{lem:grid}
Let $B$ be a block of the grid $P$. We have the following properties.  

	\begin{enumerate}[(i)]
		\item\label{adjrow} Adjacent entries in a row differ by at most $1$. 
	 	\item\label{adjcolumn} Adjacent entries in a column within the same row group differ by at most $2b$.  
		\item\label{distcolumn} Entries in non-adjacent row groups differ by more than $b$.
		\item\label{distrepr} Let $r_j$ be the representative of row group $j$. Then, all entries in the first row of row group $j-1$ and below %$r_1,\ldots,r_{j-2}$ 
		have values smaller than $r_j$ and all entries in row group $j+1$ and above have values greater than $r_j$.
	\end{enumerate}
\end{lemma}
\begin{proof}
\ref{adjrow} At any time at most one segment can start or end, which can only change the prefix sums in a column by $\pm 1$. \ref{adjcolumn} We have that adjacent entries in the leftmost column of the same row group differ by at most $b$. By \ref{adjrow} going left-to-right this difference can increase by at most $1$ in each column. Since $B$ has $b$ columns the difference can be at most $2b$. \ref{distcolumn} Any two entries in the leftmost column in two non-adjacent row groups differ by more than $2b$. Each column contains at most one update and each update can reduce this difference by no more than $1$. Hence, entries in non-adjacent row groups must differ by more than $b$. 
\ref{distrepr} The difference between $r_j$ and $r_{j-1}$ is more than $b$. Consider the first row in row group $j-1$. Since $B$ has $b$ columns it follows from \ref{adjrow} that any entry in this row has value at most $r_{j-1}+b<r_j$. Since the grid contains prefix sums, the values in a column are non-decreasing. Thus, all entries below row group $j-1$ have values smaller $r_j$. Symmetrically,  all entries in the first row of row group $j+1$ have value at least $r_{j+1}-b > r_j$.
\end{proof}

\subsection{Data Structure}
We store several data structures to represent $P$ and support queries. For each block $B$ we store the following. 

\begin{itemize}
\item A predecessor data structure on the representatives of $B$. We use the fusion node structure for constant time predecessor queries on sets of polylogaritmic size due to Fredman and Willard~\cite{FW1993,FW1994}. Since there are at most $\Delta = O(\log^\varepsilon n)$ representatives,  this structure supports queries in constant time and uses  $O(\Delta \log n) = O(b)$ bits of space. 

\item For each cell $C$, we store the leftmost column of the normalized cell $\hat{C}$. By Lemma~\ref{lem:grid}~\ref{adjrow} the first entry in the leftmost column differs from the representative $r$ by at most $b$. By Lemma~\ref{lem:grid}~\ref{adjcolumn} and since the height of $C$ is at most $\Delta$, the remaining entries in the leftmost column differ by at most $2b(\Delta - 1) + b = O(b\Delta)$. We have $\log n$ column groups in $B$ and thus the total height of all cells in $B$ is $\Delta \log n = b$. Therefore,  we can encode all leftmost columns in $O(b \log (b\Delta)) = O(b \log \log n)$ bits.  

\item For each column in $B$ we store the difference from the previous column. We encode this as the number of the slab containing the update and a single bit indicating if the update is the start or end of a segment. This uses $b \ceil{\log \Delta} + 1 = O(b \log \log n)$ bits. 
\end{itemize}
Combined we use $O(b \log \log n)$ bits for a block and thus $O(\frac{n}{b}b \log \log n) = O(n \log \log n)$ bits in total for $P$. 

We will use our data structure to efficiently construct a compact encoding for any normalized cell $\hat{C}$. To do so, we combine the encoding of leftmost column of $\hat{C}$ and the encoding of the column differences/updates in the cell in left to right order.  
%Since the width of $\hat{C}$ is $\Delta$ and the height of $\hat{C}$ is at most $\Delta$ this encoding uses at most $O(\Delta \log (b\Delta)  + \Delta \log \Delta) = O(\Delta \log \log n)$ bits. 

We will use tabulation to support the following queries on normalized cells. Given a normalized cell $\hat{C}$ and integers $i$ and $j$, define

\begin{itemize}
	\item $\Access(\hat{C}, i, j)$: return $\hat{C}(i, j)$. 
	\item $\Predecessor(\hat{C},i,j)$: return the smallest $k$ such that $\hat{C}(i,k) \geq j$. 
\end{itemize}
We construct a single global table for each of the queries. The height of $\hat{C}$ is at most $\Delta$, and by the argument above we can encode the leftmost column of $\hat{C}$ with $O(\Delta \log (b\Delta))$ bits. The rest of the columns are encoded by their difference from the previous column. Since the width of $\hat{C}$ is $\Delta$ this uses $O(\Delta \log \Delta)$ bits. Thus the encoding of $\hat{C}$  uses at most $O(\Delta \log (b\Delta)  + \Delta \log \Delta) = O(\Delta \log \log n)$ bits.
%Since the width of $\hat{C}$ is $\Delta$ and the height of $\hat{C}$ is at most $\Delta$ the encoding of $\hat{C}$  uses at most $O(\Delta \log (b\Delta)  + \Delta \log \Delta) = O(\Delta \log \log n)$ bits. 
For $\Access$ we encode %$\hat{C}$ in $O(\Delta \log \log n)$ bits as discussed above, 
the indices $i$ and $j$ using $O(\log \Delta)$ bits, and the answer in $O(\log (b\Delta))$ bits. Thus the total length of the encoding for an $\Access$ query is $O(\Delta \log \log n) + \log \Delta + \log (b\Delta)$ bits.  Hence, we can support $\Access$ in constant time with a table of size $2^{O(\Delta \log \log n + \log \Delta + \log (b\Delta))} = 2^{O(\log^\varepsilon n \log \log n)} = o(n)$ bits (recall that $\varepsilon < 1$). We encode $\Predecessor$ similarly except that the answer to the query can now be encoded in only $O(\log \Delta)$ bits. The total size the entire structure is $O(n \log \log n)$ bits.

\subsection{Supporting Queries}
We show how to implement $\SlabSum(i,j)$ and $\SlabSelect(i,j)$ in constant time. For both queries we find the block of $P$ containing column $i$ and the column group in the block corresponding to $i$. Since the blocks and column groups are evenly spaced this takes constant time. Let $B$ be the block and let $r_1, \ldots, r_m$ be the sequence of representatives in $B$ in increasing $y$-order. We then compute the predecessor $r_\ell$ of $j$ among the representatives in constant time using the fusion node structure. This identifies the cell $C_\ell$ containing entry $(i,j)$. To answer $\SlabSum(i,j)$, we compute the position $(i', j')$ in $C_\ell$ corresponding to $(i, j)$ and then compute the answer  as 
\[
\Access(\hat{C_\ell},i', j') + r_\ell. 
\]
This correctly returns the value of $C(i,j)$ since $\hat{C_\ell}$ is normalized wrt $r_\ell$.

%Since the grid contains prefix sums the values in a column are non-decreasing. Thus all entries below row group $j-1$ has values smaller $r_j$. 
To answer $\SlabSelect(i,j)$, we also consider the adjacent cells above and below $C_\ell$, denoted $C_{\ell - 1}$ and  $C_{\ell + 1}$, respectively. Since $r_\ell$ is the predecessor of $j$ we have that $r_\ell \leq j < r_{\ell+1}$. By Lemma~\ref{lem:grid}~\ref{distcolumn}, entries in row groups below row group $\ell-1$ have values smaller than $r_\ell$ and entries in row groups above row group $\ell+1$ have values greater than $r_{\ell+1}$. %in non-adjacent row groups differ by more than $b$ and 
Hence, the entry in column $i$ containing the predecessor of $j$ must be either in $C_{\ell - 1}$, $C_{\ell}$, or $C_{\ell + 1}$. We can determine the correct cell in constant time using $\SlabSum$ queries on the topmost row of each of these cells. The correct cell is the lowest of these for which the $\SlabSum$ query returns a value of at least $j$. Let $C$ denote the correct cell and let $j''$ be the topmost row in $B$ in the row group immediately below $C$. We compute the answer as 
\[
\Predecessor(\hat{C}, i', j- \SlabSum(i, j'')) + j''. 
\]   
Both queries take constant time. In summary we have shown the following result. 

\begin{lemma}\label{lem:slabselection}
	Let $L$ be a set of $n$ segments partitioned into $O(\log^\epsilon n)$ horizontal slabs. Then, we can solve the slab selection problem using $O(n \log \log n)$ bits of space and constant query time. 
\end{lemma}

\section{Segment Selection}\label{sec:segmentselection}
We now show how to solve segment selection in $O(n)$ space and $O(\log n/ \log \log n)$ query time. 
In addition to our slab selection data structure from Section~\ref{sec:slabselection}, we will also 
need a compact representation of strings that supports \emph{rank} and \emph{select queries} on polylogarithmic sized alphabets. Let $A$ be a string of length $n$ over an alphabet $[1,\sigma]$, and define the following queries: 

\begin{itemize}
	\item $\Rank(A, i, j)$: return the number of occurrences of $j$ in $X[1,i]$,
	\item $\Select(A, i, j)$: return the position of the $i$th occurrence of character $j$. 
\end{itemize}
Supporting $\Rank$ and $\Select$ on polylogarithmic sized alphabets is a well-studied problem, see e.g., \cite{FMMN2007, GGV2003,BN2015, GRR2008, GMR2006, BHMR2007, BCGNN2014}. Most of this work focuses on achieving constant time using succinct or compressed space. For our purposes we only need the following standard result which follows immediately from the above mentioned results. 
\begin{lemma}\label{lem:rankselect}
	Let $S$ be a string of length $n$ from an alphabet of size $\sigma = O(\polylog n)$. Then, we can represent $S$ in $O(n\log \sigma)$ bits and support $\Rank$ and $\Select$ queries in $O(1)$ time. 
\end{lemma}
Next, we describe our data structure. Let $L$ be a set of $n$ segments. We assume that $L$ is given in "rank space" as in the previous section. Otherwise, we can always convert $L$ into this representation by standard rank reduction techniques. %sort by $x$-coordinate, break ties using $y$-coordinate, all deletion before insertions at point $i$.  
Let $\Delta = \log^\varepsilon n$, where $0 < \varepsilon < 1$. We construct a balanced tree $R$ with degree $\Delta$ that stores the segments in $L$ in the leaves in sorted $y$-order. The height of $R$ is $O(\log_\Delta n) = O(\log n/\log \log n)$.

We introduce some helpful notation. Let $v$ be an internal node with children $v_1, \ldots, v_\Delta$. The subtree rooted at $v$ is denoted $R_v$, and the set of segments below $v$ is denoted $L_v$. We let $n_v = |L_v|$. 
The endpoints in $L_v$ are "rank reduced" to a grid of size $4n_v \times n_v$ in the following way. For an endpoint $p=(x,y)$ let $r_x$ and $r_y$ denote the rank of $p$ when the endpoints are sorted by $x$-order and $y$-order, respectively. Then $p'=(2r_x-1,r_y)$. Let $L_v'$ denote the set of rank reduced segments. % and let $A_v$ denote the set of endpoints of segments in $L_v'$ sorted by $x$-coordinates. 
The \emph{slab} of $v$, denoted $\slab(v)$,  is the narrowest infinite horizontal band containing $L_v'$. We number the slabs in increasing $y$-order. We partition the segments of $L_v'$ into $\slab(v_1),\ldots, \slab(v_\Delta)$.  
%{\bf Inge: "spaced rank" reduction here. Need to reduce problem to universe $\{1,4n_v\}$. We need a notation for the rank reduced instance. Maybe just $L_v$?}
At each internal node $v$ we store the following: 

\begin{itemize}
	\item 
	A string $E_v$ of length $4n_v$ that, for each endpoint in $L_v'$ in $x$-order, stores the slab containing it interleaved with $0$'s. More precisely, $E_v[i]$ is the number of the slab that contains the endpoint with $x$-coordinate $i$ if $i$ is odd, and $0$ if $i$ is even.
	
	%More precisely, $E_v[2i-1]$ is the number of the slab that contains the $i$th endpoint in $A_v$. We set $E[i]=0$ for all even $i$.
	%A string $E_v$ of length $4n_v$ that stores for each time $i\in\{1,\ldots,4n\}$ the number of the slab containing the endpoint from $L_v'$ with $x$-coordinate $i$. If no such endpoint exists we set $E_v[i]=0$. Thus $E_v[i]$ will be $0$ for all even $i$. %for each endpoint in $L_v$ in $x$-order the slab containing the endpoint, that is, $E_v[i]$ is the number of the slab that contains the endpoint at time $i$ in $L_v$. 
	We represent $E_v$ as a rank/select structure according to Lemma~\ref{lem:rankselect}. Since $E_v$ is a string of length $4n_v$ over an alphabet of size $\Delta$ we use $O(n_v\log \Delta) = O(n_v \log \log n)$ bits of space and support $\Rank$ and $\Select$ queries in constant time.
	%\item A string $A_v$ of length $2n_v$ that stores for each endpoint in $L_v$ in $x$-order if it is a starting or ending point. More precisely, $A_v[i]=0$ if the endpoint in $L_v$ with $x$-coordinate $i$ is a starting point and $1$ otherwise. 
	\item A slab selection structure according to  Lemma~\ref{lem:slabselection} on $L_v'$ with slabs $\slab(v_1),\ldots, \slab(v_\Delta)$. %For the slab selection structure the endpoints are "rank reduced" to a grid of size $4n_v \times n_v$ in the following way. For an endpoint $p=(x,y)\in A_v$ let $r_x$ and $r_y$ denote the rank of $p$ when the endpoints are sorted by $x$-order and $y$-order, respectively. Then $p'=(2r_x-1,r_y)$.
	
	The slab selection structure uses $O(n_v \log \log n)$ bits of space and supports queries in constant time.
\end{itemize}
\begin{figure}[t]
\centering
  \includegraphics[scale=0.53]{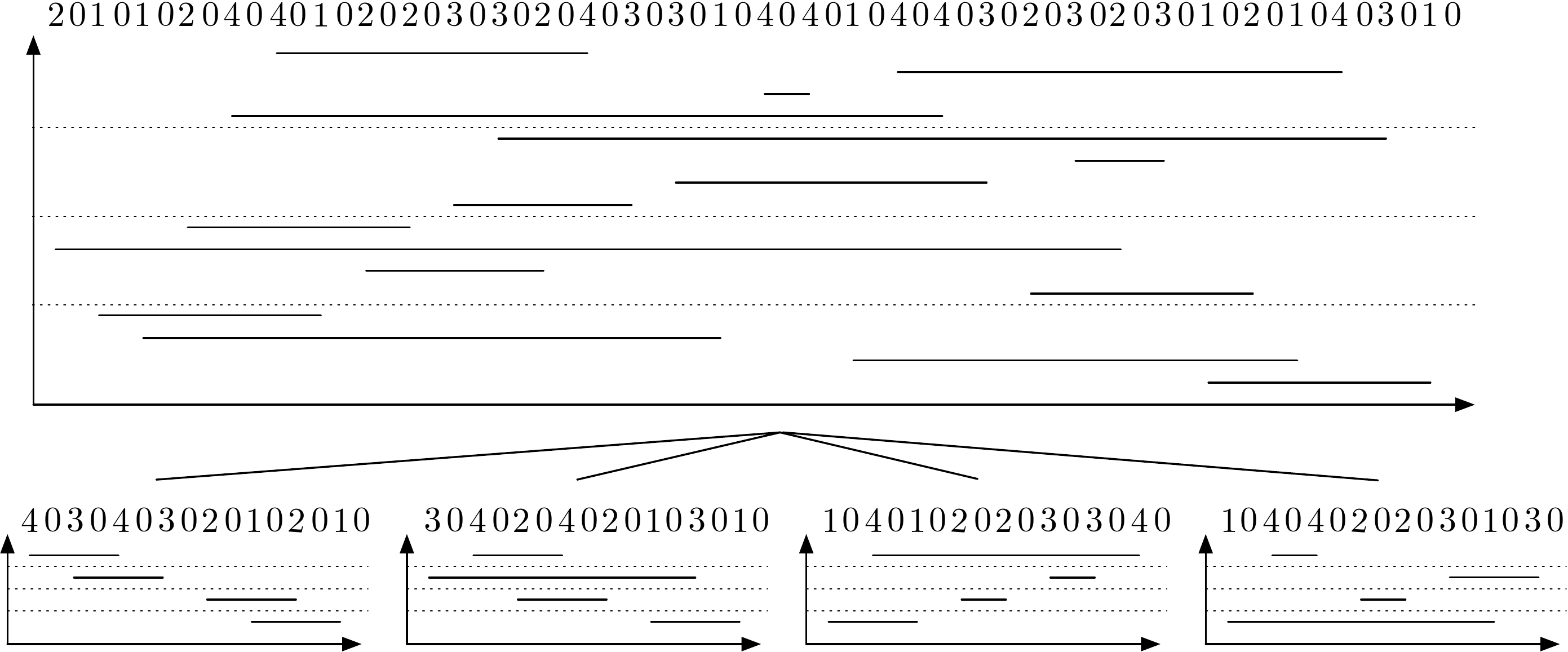}
  \caption{\label{fig:recursion} A node and its 4 children in the tree $R$ corresponding to a partition of segments into $\Delta = 4$ slabs. The string  $E_v$ is shown at each node. 
  }   
\end{figure}
See Figure~\ref{fig:recursion}.
At node $v$ we use $O(n_v \log \log n)$ bits. Since each segment appears in $O(\log n/\log \log n)$ structures the total space is $O(n \frac{\log n}{\log \log n } \log \log n) = O(n\log n)$ bits. 

To answer a $\SegmentSelect(i,j)$ query we perform a top-down search in $R$ starting at the root and ending at the leaf containing the $j$th segment that intersects the vertical line at time $i$. To guide the navigation, we compute \emph{local parameters} $i_v$ and $j_v$ at each node $v$, such that $i_v$ is the time in $L_v$ that corresponds to time $i$ in $E$, and $j_v$ is the segment in $L_v$ that corresponds to segment $j$ in $L$. At the root $r$, we have $i_r = i$ and $j_r = j$. Consider an internal node $v$ with children $v_1, \ldots, v_\Delta$ during the traversal. Given the local parameters $i_v$ and $j_v$ we compute the child to continue the search in and new local parameters. 
We first compute the slab containing the $j$th segment as 
\[k = \SlabSelect(v, i_v, j_v) \qquad \textrm{ and } \qquad j_{v_k} = j_v - \SlabSum(v, i_v, k-1)\;.\] 
Thus, the search should continue in child $v_k$, and we subtract the number of segments in the previous slabs from $j_v$ to get $j_{v_k}$. %ESA We then compute 
%ESA\[
%j_{v_k} = j_v - \SlabSum(v, i_v, k-1)
%\]
To compute $i_{v_k}$ we first compute 
%\[
$
r_k = \Rank(E_v, i_v, k)
$. 
%\]
Since $i_v$ might not be a point in $L_{v_k}$
we then set
 %$\Select(E_v, \Rank(E_v, i_v, k), k)$
%Now $r_k$ is the number of endpoints of segments in slab $k$ up to time $i_v$. We now set 
\[
i_{v_k} = 
\begin{cases}
2r_k - 1& \textrm{if } E_v[i_v] = k\\
2r_k  & \textrm{otherwise }
\end{cases}
\]
By Lemma~\ref{lem:slabselection} and Lemma~\ref{lem:rankselect} each of the above steps takes constant time and hence the total time is $O(\log n/\log \log n)$. We have the following property of $i_{v_k}$.
%This ensures that the segments from slab $k$ in $L_v'$ that are intersected by $i_v$ are the same as the segments intersected by $i_{v_k}$ in $L_{v_k}'$.

\begin{lemma}\label{lem:rankreduction}
The segments from slab $k$ in $L_v'$ that are intersected by $i_v$ are the same as the segments intersected by $i_{v_k}$ in $L_{v_k}'$. 
\end{lemma}
\begin{proof}
Let $I_v$ denote the set of segments from slab $k$ in $L_v'$ intersected by $i_v$ and let $I_{v_k}$ denote the set of segments in $L_{v_k}'$ intersected by $i_{v_k}$. Let $s$ a segment from $L_{v_k}$ with $x$-coordinates $(x_1,x_2)$ in $L_v'$ and $(x_1', x_2')$ in $L_{v_k}'$.  Define $r_i = \Rank(E_v,x_i,k)$. From the definition of the rank reduction we have $x_i' = 2\cdot r_i-1$. We will show that $s \in I_v$ iff $s \in I_{v_k}$.

First assume $s \in I_v$. Then $x_1\leq i_v\leq x_2$, which implies that  $\Rank(E_v,x_1,k) \leq \Rank(E_v,i_v,k)\leq \Rank(E_v,x_2,k)$, that is $r_1\leq r_k\leq r_2$. 
We need to prove that $x_1'\leq i_{v_k}\leq x_2'$. If $E_v[i_v]=k$, i.e., $i_v$ is an endpoint of a segment in slab $k$ then it immediately follows that $x_i'=2r_1-1 \leq 2r_k-1 = i_{v_k}$ and similarly that $i_{v_k}\leq x_2'$.  If $E_v[i_v]\neq k$ then $i_{v_k}$ is not an endpoint in $L_{v_k}$ and thus $x_1< i_v< x_2$. This implies that $r_1\leq r_k<r_2$. We have $r_1=r_k$ in the case where $x_1$ is the rightmost endpoint in slab $k$ smaller than $i_v$. It follows immediately that $2r_1-1 \leq 2r_k-1< 2 r_k < 2 r_2$, and therefore $x_1'< i_{v_k}< x_2'$.

Assume $s \in I_{v_k}$. Then $x_1'\leq i_{v_k}\leq x_2'$ and we want to prove that $x_1\leq i_v\leq x_2$. We have $2r_1-1\leq i_{v_k}\leq 2r_2-1$.  We will first show that $r_1\leq r_k\leq r_2$. There are two cases. If $E[i_v]=k$ then $i_{v_k} = 2r_k -1$ and it follows immediately that $r_1\leq r_k\leq r_2$.  If $E[i_v]\neq k$ then $i_{v_k} = 2r_k$ and thus $2r_1-1\leq i_{v_k}\leq 2r_2-1$ implies $2r_1-1< 2r_k < 2r_2-1$ which again implies that $r_1\leq r_k\leq r_2$. By definition of $\Rank$ we have that $\Rank(E_v,x_1,k) \leq \Rank(E_v,i_v,k)\leq \Rank(E_v,x_2,k)$ implies $x_1\leq i_v\leq x_2$.
\end{proof}
%If $p$ is the ending point of a segment we need to add one to $j_{v}$, since we then have one more segment at time $p$ than at time $i_v$. 
%
%If $p$ is the ending point of a segment we need to add one to $j_{v}$, since we then have one more segment at time $p$ than at time $i_v$. 
%and $i_{v_k}$ is therefore the position in $E_{v_k}$ of the predecessor $p$ of $i_v$ in set of endpoints of $L_{v_k}$. 
%

%By Lemma~\ref{lem:slabselection} and Lemma~\ref{lem:rankselect} each of the above steps takes constant time and hence the total time is $O(\log n/\log \log n)$. 
In summary, this proves Theorem~\ref{thm:segmentselection}.  

Combined with the reduction in Lemma~\ref{lem:reduction} we obtain a linear space and $O(\log n/ \log \log n)$ time solution for the random access problem. To show  Theorem~\ref{thm:mainupperbound} it only remains to show how to report a substring of length $\ell$ in $O(\log n/\log \log n + \ell)$ time. To do so we build the \emph{hive graph} of Chazelle~\cite{Chazelle1986} on the segments. This uses $O(n)$ space and allows us to traverse the segments through the vertical line at time $i$ above a given segment in sorted order in constant time per reported segment. To report a substring of length $\ell$ we simply 
perform the corresponding segment selection and traverse the $\ell$ segments above. By Lemma~\ref{lem:segments} this gives us the correct substring. This uses $O(\log n/ \log \log n + \ell)$ time. This completes the proof of Theorem~\ref{thm:mainupperbound}.

Finally, we show how to construct the random access data structure of  Theorem~\ref{thm:mainupperbound} in $O(n \log n)$ time. Given a version tree $T$ with $n$ nodes it is straightforward to construct the corresponding instance of the segment selection problem $L$ as described in Section~\ref{sec:reduction} in $O(n)$ time in a single traversal of $T$. We then construct tree $R$ over the segments in $L$  recursively. At each node $v$ we build the slab selection data structure from Section~\ref{sec:slabselection} consisting of $n_v$ segments. To do so, we construct the grid,  the predecessor data structure, and the compact encoding in $O(n_v)$ time. The global tables for the normalized cells need only be constructed once in $O(n)$ total time. Furthermore, we also need to build the rank/select data structure from Lemma~\ref{lem:rankselect}. This can also be done in $O(n_v)$ time, and hence constructing these data structures on all nodes in $R$ takes $O(n \log_\Delta n) = O(n \log n/ \log \log n)$ time. Finally, constructing the hive graph can be done in $O(n \log n)$ time~\cite{Chazelle1986}.

\section{Lower Bounds}\label{sec:lowerbound}
\begin{figure}[t]
\centering
  \includegraphics[scale=0.6]{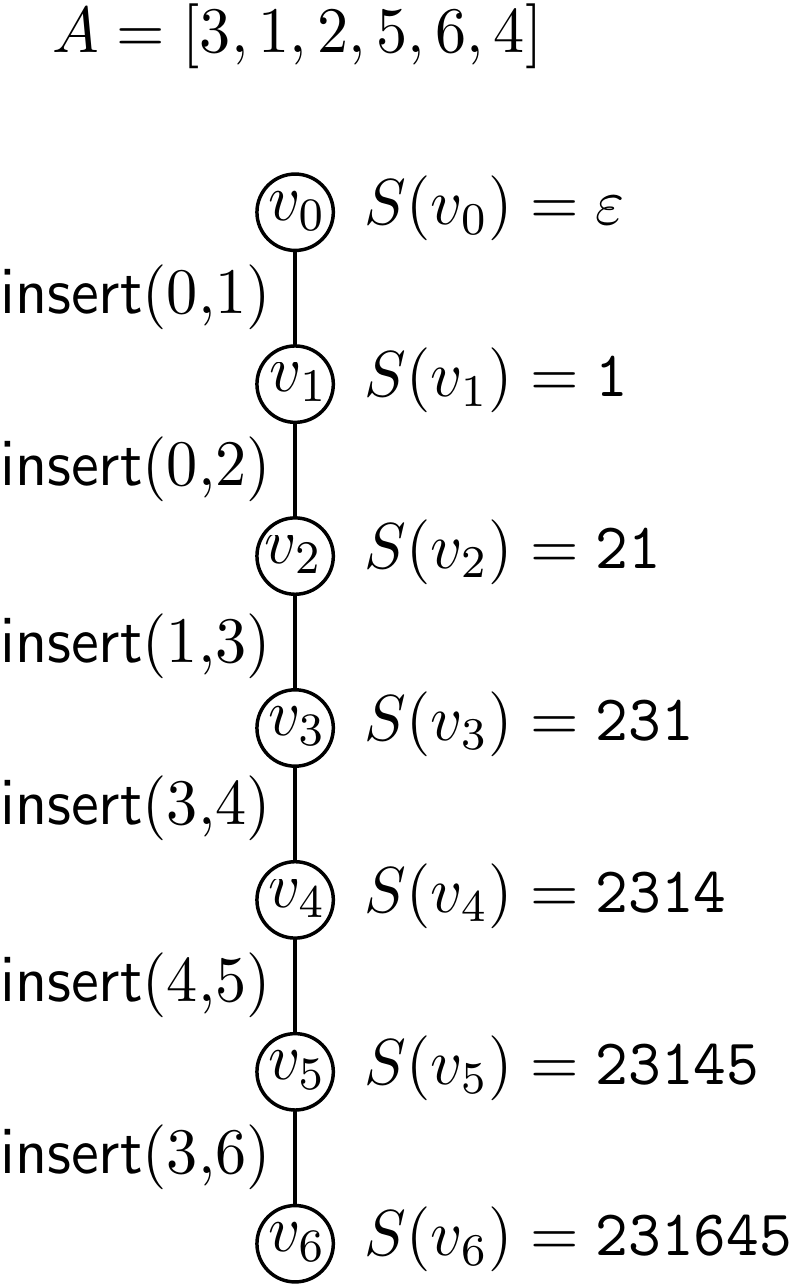}
  \caption{\label{fig:lowerbound} The corresponding random access instance for an array $A = [3,1,2,5,6,4]$.}   
\end{figure}

We now prove the lower bounds in Theorems~\ref{thm:mainlowerbound} and~\ref{thm:segmentselectionlowerbound} for random access and segment selection, respectively. For the random access problem we show a reduction from the following problem: Let $A$ be an array of $n$ unique integers. The \emph{prefix selection problem} is to preprocess $A$ to support prefix selection queries, that is, given integers $i$ and $j$ report the $j$th smallest integer in the subarray $A[1..i]$.
\begin{lemma}[J{\o}rgensen and Larsen~\cite{JL2011}] Any data structure that uses $n\log^{O(1)} n$ space on an input array of size $n$ needs $\Omega(\log n/ \log \log n)$ time to support prefix selection queries. 	
\end{lemma}
Given an input array $A$ to the prefix selection problem, we construct an instance $T$ of the random access problem. Our reduction allows any prefix selection query on $A$ to be answered by a single random access query on $T$. The reduction works even when $T$ is a path without any deletions.

Let $A$ be an array of length $n$ consisting of unique integers in $\{1, \ldots, n\}$. Our instance $T$ is a path of $n+1$ nodes $v_0, \ldots, v_n$ rooted at $v_0$. See Figure~\ref{fig:lowerbound}. Edge $(v_{i-1}, v_i)$ is labeled by $\Insert(r_i, i)$, where $r_i$ is the number of entries in $A[1..i]$ that are smaller $A[i]$. We have that $S(v_i)$ is permutation of indices in $\{1,..,i\}$ corresponding to the sorted order of $A[1..i]$, that is, $A[S(v_i)[1]] < \cdots < A[S(v_i)[i]]$. In particular, $S(v_i)[j]$ is the index of the $j$th smallest integer in $A[1..i]$. Hence, we can answer a prefix selection query $\PrefixSelect(i,j)$ by computing $\Access(v_i, j)$. This completes the proof of Theorem~\ref{thm:mainlowerbound}.

For the segment selection problem, we note that our reduction in Lemma~\ref{lem:reduction} combined with Theorem~\ref{thm:mainlowerbound} directly implies Theorem~\ref{thm:segmentselectionlowerbound}.

\section{Conclusion and Open Problems}\label{sec:conclusion}
We have initiated the study of persistent strings for storing and accessing compressed collections of similar strings. We have shown how to store a persistent string in linear space with optimal random access time. An interesting open problem is to make our solution \emph{dynamic} by supporting insertion of new nodes in the version tree (representing new strings added to the collection). Another open problem is to improve our straightforward $O(n\log n /\log \log n)$ preprocessing time to optimal $O(n)$.

\section{Acknowledgments}
We thank Jesper Jansson for pointing out segment selection as a problem of independent interest and the anonymous reviewers for their helpful comments that improved the presentation of earlier versions of this paper.  

\bibliography{master}

\end{document}